\documentclass[reqno]{amsart}
\usepackage{amsmath,amssymb,cite}
\usepackage[mathscr]{euscript}
\usepackage{graphics,epsfig,color}
\theoremstyle{plain}
\newtheorem{theorem}{Theorem}[section]
\newtheorem{proposition}[theorem]{Proposition}
\newtheorem{lemma}[theorem]{Lemma}

\theoremstyle{definition}

\newtheorem{assumption}[theorem]{Assumption}

\theoremstyle{remark}
\newtheorem{remark}[theorem]{Remark}

\numberwithin{equation}{section}
\numberwithin{theorem}{section}

\newcommand{\mc}[1]{{\mathcal #1}}

\newcommand{\bs}[1]{{\boldsymbol #1}}
\newcommand{\bb}[1]{{\mathbb #1}}

\newcommand{\rmd}{\mathrm{d}}
\newcommand{\eps}{\varepsilon}

\def\const{({\rm const.})\,}

\newcommand{\supp}{\mathop{\rm supp}\nolimits}
\renewcommand{\div}{\mathop{\rm div}\nolimits}

\title[Time evolution of concentrated vortex rings]{Time evolution of concentrated vortex rings}

\author[P.\ Butt\`a]{Paolo Butt\`a}
\address{Paolo Butt\`a\hfill\break \indent
Dipartimento di Matematica, 
Sapienza Universit\`a di Roma,
\hfill\break \indent
P.le Aldo Moro 5, 00185 Roma, Italy}
\email{butta@mat.uniroma1.it}
\author[C.\ Marchioro]{Carlo Marchioro}
\address{Carlo Marchioro \hfill\break \indent
Dipartimento di Matematica, 
Sapienza Universit\`a di Roma,
\hfill\break \indent
P.le Aldo Moro 5, 00185 Roma, Italy \hfill\break \indent and \hfill\break \indent  
International Research Center M\&MOCS, Universit\`a di L'Aquila \hfill\break \indent Palazzo Caetani, 04012 Cisterna di Latina (LT), Italy}
\email{marchior@mat.uniroma1.it}

\begin{document}

\begin{abstract}
We study the time evolution of an incompressible fluid with axisymmetry without swirl when the vorticity is sharply concentrated. In particular, we consider $N$ disjoint vortex rings of size $\eps$ and intensity of the order of $|\log\eps|^{ -1}$. We show that in the limit $\eps\to 0$, when the density of vorticity becomes very large, the movement of each vortex ring converges to a simple translation, at least for a small but positive time.
\end{abstract}



\maketitle
\thispagestyle{empty}

\section{Introduction and main result}
\label{sec:1}
In this paper we study the time evolution of an incompressible fluid with constant density when the axisymmetry without swirl is present and the vorticity is  sharply concentrated. A similar situation has been largely investigated in the case of planar symmetry, a short summary and references on this subject are given in Appendix \ref{app:a}. Here, we show how several  not trivial problems arise in presence of this different symmetry.

The evolution of an incompressible inviscid fluid of unitary density filling the whole space $\bb R^3$ is governed by the Euler equations,
\begin{equation}
\label{eule}
\partial_t  \bs u + (\bs u\cdot \nabla) \bs u  =  -\nabla p \,, \qquad \nabla \cdot \bs u = 0 \,,
\end{equation}
where $\bs u = \bs u(\bs\xi,t)$ and $p=p(\bs\xi,t)$ are the velocity and pressure respectively, $\bs\xi = (\xi_1,\xi_2,\xi_3)$ denotes a point in $\bb R^3$, and $t\in \bb R_+$ is the time. More precisely, the evolution is determined by the Cauchy problem associated to \eqref{eule} with assigned boundary conditions. In what follows, we always assume that $u$ decays at infinity. Introducing the vorticity $\bs\omega$, defined by
\begin{equation}
\label{vort}
\bs\omega = \nabla \wedge \bs u \,,
\end{equation}
this assumption implies that the velocity $u$ can be reconstructed from $\bs\omega$ as
\begin{equation}
\label{u-vort}
\bs u(\bs\xi,t) = - \frac{1}{4\pi} \int\! \rmd \bs\eta \, \frac{(\bs\xi-\bs\eta) \wedge \bs\omega(\bs\eta,t)}{|\bs\xi-\bs\eta|^3} \,.
\end{equation}
On the other hand, by \eqref{eule} and \eqref{vort}, the vorticity evolves governed by the equation
\begin{equation}
\label{vorteq}
\partial_t \bs\omega + (\bs u\cdot \nabla) \bs\omega  = (\bs \omega\cdot \nabla) \bs u  \,,
\end{equation}
so that Eqs.\ \eqref{u-vort} and \eqref{vorteq} give a formulation of the Euler equations (with velocity decaying at infinity) in terms of the vorticity $\bs\omega$.

Denoting by $(z,r,\theta)$ the cylindrical coordinates, we recall that the vector field $\bs F$ of cylindrical components $(F_z, F_r, F_\theta)$ is called axisymmetric without swirl if $F_\theta=0$ and $F_z$ and $F_r$ are independent of $\theta$. 

We observe that the axisymmetry is preserved by the evolution \eqref{eule}. Moreover, when restricted to axisymmetric velocity fields $\bs u(\bs\xi,t) = (u_z(z,r,t), u_r(z,r,t), 0)$, Eqs.\ \eqref{vort} and \eqref{vorteq} reduces to
\begin{equation}
\label{omega}
\bs\omega = (0,0,\omega_\theta) = (0,0,\partial_z u_r - \partial_r u_z)
\end{equation}
and, denoting henceforth $\omega_\theta$ by $\omega$, 
\begin{equation}
\label{omeq}
\partial_t \omega + (u_z\partial_z + u_r\partial_r) \omega - \frac{u_r\omega}r = 0 \,.
\end{equation}
We also notice that the solenoidal condition $\nabla \cdot \bs u = 0$ reads
\[
\partial_z(ru_z) + \partial_r(ru_r) = 0\,.
\]
Finally, by \eqref{u-vort}, $u_z = u_z(z,r,t)$ and $u_r=u_r(z,r,t)$ are given by
\begin{align}
\label{uz}
u_z & = - \frac{1}{2\pi} \int\! \rmd z' \!\int_0^\infty\! r' \rmd r' \! \int_0^\pi\!\rmd \theta \, \frac{\omega(z',r',t) (r\cos\theta - r')}{[(z-z')^2 + (r-r')^2 + 2rr'(1-\cos\theta)]^{3/2}} \,,
\\ \label{ur}
u_r & = \frac{1}{2\pi} \int\! \rmd z' \!\int_0^\infty\! r' \rmd r' \! \int_0^\pi\!\rmd \theta \, \frac{\omega(z',r',t) (z - z')}{[(z-z')^2 + (r-r')^2 + 2rr'(1-\cos\theta)]^{3/2}} \,.
\end{align}
Hence, the axisymmetric solutions to the Euler equations are the solutions to Eqs.\ \eqref{omeq}, \eqref{uz}, and \eqref{ur}.

We further observe that Eq.~\eqref{omeq} means that the quantity $\omega/r$ remains constant along the flow generated by the velocity field, i.e., 
\begin{equation}
\label{cons-omr}
\frac{\omega(z(t),r(t),t)}{r(t)} = \frac{\omega(z(0),r(0),0)}{r(0)} \,,
\end{equation}
with $(z(t),r(t))$ solution to
\begin{equation}
\label{eqchar}
\dot z(t) = u_z(z(t),r(t),t) \,, \qquad \dot r(t) = u_r(z(t),r(t),t) \,.
\end{equation}

It is possible to consider non-smooth initial data by assuming  \eqref{uz}, \eqref{ur}, \eqref{cons-omr}, and \eqref{eqchar} as a weak formulation of the Euler equations in the framework of  axisymmetric solutions. An equivalent weak formulation is obtained from \eqref{omeq} by a formal integration by parts,
\begin{equation}
\label{weq}
\frac{\rmd}{\rmd t} \omega_t[f] = \omega_t[u_z\partial_z f + u_r\partial_r f + \partial_t f] \,,
\end{equation}
where $f = f(z,r,t)$ is any bounded smooth test function and 
\[
\omega_t[f] := \int\! \rmd z \!\int_0^\infty\! \rmd r \, \omega(z,r,t) f(z,r,t) \,.
\]
In particular, global (in time) existence and uniqueness of the weak solution to the associated Cauchy problem holds when the initial vorticity $\omega^0(z,r) := \omega(z,r,0)$ is a bounded function with compact support contained in the open half-plane $\Pi:= \{(z,r) \colon r>0\}$, see, e.g., \cite[Appendix]{CS}. Moreover, the support of the vorticity remains in the open half-plane $\Pi$ at any time. 

A point of the half-plane $\Pi$ corresponds to a circumference in the whole space and the special class of axisymmetric solutions without swirl are called sometime \textit{smoke rings}, because there exist particular solutions whose shape remains constant in time (the so-called steady vortex ring) and translate in the $z$-direction with constant speed (the propagation velocity). The existence of such solutions is an old issue, see \cite{FrB74,AmS89} for a rigorous proof by means of variational methods. We also quote the review \cite{ShL92} as a bibliographic reference in physical literature on axisymmetric solutions without swirl.

Here, we are interested in the special class of initial data when the vorticity is sharply concentrated.

We first discuss the case when the concentration occurs around a single point $(z_0,r_0)$. Denoting by $\Sigma((z,r)|\rho)$ the open disk of center $(z,r)$ and radius $\rho$, we assume that $\omega_\eps(z,r,0)$ is non-negative and that its support is inside the disk $\Sigma((z_0,r_0)|\eps)$, with $\eps>0$ a small parameter. We also let
\begin{equation}
\label{Momega}
N_\eps = \int\!\rmd z \!\int_0^\infty\!\rmd r\,\omega_\eps(z,r,0)\,.
\end{equation}
The case of a steady vortex ring is analyzed in \cite{Fr70}, where the propagation velocity is shown to be approximately equal to $N_\eps|\log \eps|/(4\pi r_0)$ if $\eps$ is very small. We are interested in the extension of this result to \textit{any} concentrated initial data and not only to the particular case of steady vortex rings. The analysis of the latter case suggests that to obtain a finite limit propagation velocity $N_\eps$ has to be chosen of the order of $|\log \eps |^{-1}$. In fact, such analysis has been done in \cite{BCM00}, where it is proven that if there are $M,a>0$ such that, for any $\eps$ small enough,
\[
N_\eps = \frac{a}{|\log\eps|}\,,  \qquad 0 \le \omega_\eps(z,r,0) \le \frac{M}{\eps^2|\log \eps|}\,,
\]
then for any $T>0$ the evolution $\omega_\eps(z,r,t)$, $t\in [0,T]$, is concentrated in a disk $\Sigma((z_\eps(t),r_\eps(t))|D_\eps)$ of radius $D_\eps = \eps |\log \eps|$, in the sense that
\[
\lim_{\eps\to 0} |\log \eps| \int_{\Sigma((z_\eps(t),r_\eps(t))|D_\eps)}\! \rmd z\, \rmd r\, \omega_\eps(z,r,t) = a\,.
\]
Moreover,
\[
\lim_{\eps\to 0} (z_\eps(t),r_\eps(t)) = (z_0+vt,r_0)\,, \quad v = \frac{a}{4\pi r_0}\,.
\]

The previous result holds for one vortex ring alone. The aim of the present paper is to extend the analysis to the case of $N$ vortex rings. The proof is based on a not trivial improvement of the previous result on the motion of a single vortex ring. Actually, in \cite{BCM00} the main effort has been made to prove that the vorticity remains concentrated around a point, whereas small filaments of vorticity mass $m_\eps$ could go away. As $\eps \to 0$ this point has a linear motion and $m_\eps \to 0$. But now, to control the interaction among the vortex rings, we need to show that these filaments remain close to their concentration points of vorticity. We are able to prove such stronger property of the evolution during a suitable time interval, which cannot be chosen arbitrarily large but it can be fixed \textit{independent} of $\eps$ (a sort of rough localization). With respect to the planar case, where it is possible to prove a sharp localization, there is a main difference: the planar symmetry gives rise to a quasi-Lipschitz velocity field, while in the axisymmetry without swirl  the velocity field is much more singular.

We now state the main result of the paper. Given a small parameter $\eps\in (0,1)$, we consider initial data for which the vorticity is supported in $N$ disks, i.e., an initial vorticity of the form
\begin{equation}
\label{in}
\omega_\eps(z,r,0) = \sum_{i=1}^N  \omega_{i,\eps}(z,r,0)\,,
\end{equation}
where $\omega_{i,\eps}(z,r,0)$, $i=1,\ldots, N$, are functions with definite sign such that
\begin{equation}
\label{initial}
\Lambda_{i,\eps}(0) := \supp\, \omega_{i,\eps}(\cdot,0) \subset \Sigma(\zeta^i|\eps)\,,
\end{equation}
for fixed $\zeta^i = (z_i,r_i)\in \Pi$. We assume $\eps$ so small to have
\[
\overline{\Sigma(\zeta^i|\eps)} \subset\Pi\quad \forall\, i = 1,\ldots,N\,, \quad  \Sigma(\zeta^i|\eps)\cap \Sigma(\zeta^j|\eps)=\emptyset\quad \forall\, i \ne j\,.
\]
As already noticed, this implies that the solution $\omega_\eps(z,r,t)$ is defined globally in time and its support remains inside $\Pi$. In view of the axisymmetry, each component $\omega_{i,\eps}(z,r,0)$ in the right-hand side of \eqref{in} represents a  vortex ring, of width of the order of $\eps$ and radius $\zeta^i$ (which is chosen independent of $\eps$). In general, the signs of the functions $\omega_{i,\eps}(z,r,0)$ can be different among each other, and we assume that the intensity of the $i$-th ring vanishes logarithmically as $\eps\to 0$. More precisely, we assume that there are $N$ real parameters $a_1,\ldots, a_N$ such that
\begin{equation}
\label{ai}
|\log\eps| \int\!\rmd z \!\int_0^\infty\!\rmd r\, \omega_{i,\eps}(z,r,0) = a_i \quad \forall\,i=1,\ldots,N\,.
\end{equation}
We further suppose that there is $M>0$ such that
\begin{equation}
\label{Mgamma}
|\omega_{i,\eps}(z,r,0)| \le \frac{M}{\eps^2|\log\eps|} \quad \forall\, x\in\bb R^2\quad \forall\, i=1,\ldots,N\,.
\end{equation}

The decomposition \eqref{in} naturally extends to time $t>0$ because, in view of \eqref{cons-omr},
\begin{equation}
\label{in-t}
\omega_\eps(z,r,t) = \sum_{i=1}^N  \omega_{i,\eps}(z,r,t)\,,
\end{equation}
with $\omega_{i,\eps}(z,r,t)$ the time evolution of the $i$th vortex ring, i.e., such that
\begin{equation}
\label{cons-omr_ni}
\omega_{i,\eps}(z(t),r(t),t) := \frac{r(t)}{r(0)} \omega_{\eps,i}(z(0),r(0),0)\,.
\end{equation}

The main result is the content of the following theorem.

\begin{theorem}
\label{thm:1}
Assume the initial data $\omega_\eps(z,r,0)$ verify \eqref{in}, \eqref{initial}, \eqref{ai}, and \eqref{Mgamma}, and define, for $\zeta^i = (z_i,r_i)\in \Pi$, $i=1,\ldots,N$, as in \eqref{initial},
\begin{equation}
\label{free_m}
\zeta^i(t) :=  \zeta^i + \begin{pmatrix} v_i \\ 0 \end{pmatrix} t \qquad \text{with}\quad v_i =\frac{a_i}{4\pi r_i}\,.
\end{equation}
Then, for any $R>0$ such that the closed disks $\overline{\Sigma(\zeta^i|R)}$ are mutually disjoint there exist $\eps_R\in (0,1)$ and $T_R>0$ such that, for any $t\in[0,T_R]$ and $ i=1,\ldots,N$, the following holds true.

\begin{itemize}
\item[(1)] $\Lambda_{i,\eps}(t) := \supp\, \omega_{i,\eps}(\cdot,t) \subseteq \Sigma(\zeta^i(t)|R)$ for any $\eps\in (0,\eps_R]$, and the disks $\Sigma(\zeta^i(t)|R)$ are mutually disjoint.
\item[(2)] There are $\zeta^{i,\eps}(t)\in \Pi$ and $R_\eps>0$ such that
\[
\lim_{\eps\to 0}|\log\eps| \int_{\Sigma(\zeta^{i,\eps}(t)|R_\eps)}\!\rmd z\,\rmd r\, \omega_{i,\eps}(z,r,t) = a_i\quad \forall\, i=1,\ldots,N\,,
\]
with
\[
\lim_{\eps\to 0} R_\eps = 0\,, \quad \lim_{\eps\to 0} \zeta^{i,\eps}(t) = \zeta^i(t)\,.
\]
\end{itemize}
\end{theorem}

The proof of Theorem \ref{thm:1} is obtained in two separate steps. We first prove an analogous result for a ``reduced system'' which describes the motion of a single vortex ring in a suitable external time-dependent vector field. The original problem is then solved by simulating the motion of each vortex ring by means of the reduced system, in which the external field describes the force due to its interaction with the other rings.

The same approach has been used in the planar case, where a key ingredient to prove the localization property is a quite sharp a priori estimate on the moment of inertia, which is not available here because the velocity field is not a Lipschitz function. Indeed, the case of one vortex ring alone has been successfully solved by making use of the energy conservation. But this is not sufficient to prove the localization property, which appears necessary to cover the case of many vortex rings. We then use a mix of the two strategies. To overcome the lack of Lipschitz property, we apply the energy conservation to control the growth in time of the moment of inertia. This estimate allows us to build up an iterative scheme to deduce the sharp localization property. The price to pay is that the method works only if the time is not too large.

The plan of the paper is the following. In the next section we introduce the reduced model and prove Theorem \ref{thm:1} as a corollary of the analogous result for the reduced system. The proof of the latter is given in Section \ref{sec:3}. Some comments on related problems are given in Section \ref{sec:4}. As already mentioned, in Appendix \ref{app:a} we briefly recall the known result in the case of planar symmetry, while Appendix \ref{app:b} is devoted to the proof of some technical results. 

\section{Proof of the main result}
\label{sec:2}

\subsection{Preliminaries}
\label{sec:2.1}

To have a more compact notation, we rename the variables by setting
\begin{equation}
\label{nv}
x = (x_1,x_2) := (z,r)\,,
\end{equation}
so that $\Pi = \{x=(x_1,x_2) \in\bb R^2\colon x_2>0\}$.

In what follows, we tacitly extend the vorticity to a function on the whole plane $\bb R^2$ by setting $\omega_\eps(x,t) = 0$ for $x_2\le 0$, so that the equations of motion \eqref{uz}, \eqref{ur}, \eqref{cons-omr} and \eqref{eqchar} are reshaped in the following form,
\begin{equation}
\label{u=}
u(x,t) = \int\!\rmd y\, H(x,y)\, \omega_\eps(y,t)\,,
\end{equation}
\begin{equation}
\label{cons-omr_n}
\omega_\eps(x(t),t) = \frac{x_2(t)}{x_2(0)} \omega_\eps(x(0),0) \,, 
\end{equation}
\begin{equation}
\label{eqchar_n}
\dot x(t) = u(x(t),t) \,,
\end{equation}
where the kernel $H(x,y) = (H_1(x,y),H_2(x,y))$ is given by
\begin{align}
\label{H1}
H_1(x,y) & = \frac{1}{2\pi} \int_0^\pi\!\rmd \theta \, \frac{y_2(y_2 - x_2\cos\theta)}{\big[|x-y|^2 + 2x_2y_2(1-\cos\theta)\big]^{3/2}} \,,
\\ \label{H2}
H_2(x,y) & = \frac{1}{2\pi} \int_0^\pi\!\rmd \theta \, \frac{y_2 (x_1-y_1) \cos\theta}{\big[|x-y|^2 + 2x_2y_2(1-\cos\theta)\big]^{3/2}} \,.
\end{align}

As mentioned in the previous section, we first prove an analogue of Theorem \ref{thm:1} for a ``reduced system'' which describes the motion of a single vortex ring in a suitable external time-dependent vector field. This is the content of Theorem \ref{thm:2} below, whose proof is postponed to the next section.

\subsection{The reduced system}
\label{ssec:2.3}

The reduced system is defined by Eqs.\ \eqref{u=}, \eqref{cons-omr_n}, and, in place of \eqref{eqchar_n}, 
\begin{equation}
\label{eqchar_nF}
\dot x(t) = u(x(t),t) + F^\eps(x(t),t)\,.
\end{equation}
The initial datum $\omega_\eps(x,0)$ and the time dependent vector field $F^\eps$ (possibly depending on $\eps$) are assumed to satisfy the following conditions. 

\begin{assumption}
\label{ass:1}
The function $\omega_\eps(x,0)$ is non-negative (resp.~non-positive) and there is $M>0$ and $a>0$ (resp.~$a<0$) such that 
\begin{equation}
\label{MgammaF}
0 \le |\omega_\eps(x,0)| \le \frac{M}{\eps^2|\log\eps|} \quad \forall\, x\in\bb R^2\,, \qquad |\log\eps|\int\!\rmd y\, \omega_\eps(y,0) =a\,.
\end{equation}
Moreover, there exists $\zeta^* = (z_*,r_*) \in \Pi$ such that
\begin{equation}
\label{initialF}
\Lambda_\eps(0) := \supp\, \omega_\eps(\cdot,0) \subset \Sigma(\zeta^*|\eps)\,.
\end{equation}
We also assume $\eps$ so small to have
\[
\overline{\Sigma(\zeta^*|\eps)} \subset\Pi\,.
\]
Finally, $F^\eps=(F^\eps_1,F^\eps_2) \in C(\bb R^2\times [0,\infty);\bb R^2)$ is globally Lipschitz and enjoys the following properties.
\begin{itemize}
\item[(a)] The vector field $\bs F^\eps = (F^\eps_z,F^\eps_r,F^\eps_\theta) := (F^\eps_1,F^\eps_2,0)$ has zero divergence, i.e., $\partial_{x_1}(x_2 F_1) + \partial_{x_2}(x_2 F_2) = 0$.
\item[(b)] There exist $C_F, L >0$ such that, for any $\eps\in (0,1)$ and $t\ge 0$,
\[
|F^\eps(x,t)| \le \frac{C_F}{|\log\eps|}\,, \quad |F^\eps(x,t) - F^\eps(y,t)| \le \frac{L}{|\log\eps|} |x-y|\qquad \forall\,x,y\in\bb R^2\,.
\]
\end{itemize}
\end{assumption}

\begin{remark}
\label{rem:1}
Under Assumption \ref{ass:1}, the argument in \cite[Appendix]{CS} can be easily adapted to the present context to prove existence and uniqueness of solutions for the reduced problem. Moreover, the support of $\omega_\eps(\cdot,t)$ remains inside the open half-plane $\Pi$ (in particular, in Eq.~\eqref{u=} the integration is actually restricted to $x_2 > 0$). Finally, the following weak formulation holds true,
\begin{equation}
\label{weqF}
\frac{\rmd}{\rmd t} \int\!\rmd x\, \omega_\eps(x,t) f(x,t) = \int\!\rmd x\, \omega_\eps(x,t) \big[ (u+F^\eps) \cdot \nabla f + \partial_t f \big](x,t) \,,
\end{equation}
where $f = f(x,t)$ is any bounded smooth test function. We omit the details.
\end{remark}

\begin{theorem}
\label{thm:2}
Under Assumption \ref{ass:1}, let 
\begin{equation}
\label{zett}
\zeta(t) =  \zeta^* + \begin{pmatrix} v \\ 0 \end{pmatrix} t \qquad \text{with}\quad v =\frac{a}{4\pi r_*}\,.
\end{equation}
Then, for each $R_*\in (0,r_*)$ there are $T_*>0$ and $\eps_*\in (0,1)$ such that for any $t\in[0, T_*]$ the following holds true.

\begin{itemize}
\item[(1)] $\Lambda_\eps(t) := \supp\, \omega_\eps(\cdot,t) \subseteq \Sigma(\zeta(t)|R_*)$ for any $\eps\in (0,\eps_*]$.
\item[(2)] There are $\zeta^\eps(t)\in \Pi$ and $\varrho_\eps>0$ such that
\[
\lim_{\eps\to 0}|\log\eps| \int_{\Sigma(\zeta^\eps(t)|\varrho_\eps)}\!\rmd x\, \omega_\eps(x,t) = a\,,
\]
with
\[
\lim_{\eps\to 0} \varrho_\eps = 0\,, \quad \lim_{\eps\to 0} \zeta^\eps(t) = \zeta(t)\,.
\]
\end{itemize}
\end{theorem}

\subsection{Proof of Theorem \ref{thm:1}}
\label{sec:2.3}

Given $R$ as in the statement of the theorem, let 
\[
T_R' := \sup\left\{t>0\colon \min_{i\ne j}|\zeta^i(s)-\zeta^j(s)|>2R \;\; \forall\, s\in [0,t]\right\}\,,
\]
and
\[
T_R^\eps := \sup\left\{t\in [0,T_R') \colon \Lambda_{i,\eps}(s) \subseteq \Sigma(\zeta^i(s)|R/2)\;\; \forall\, s\in [0,t]\;\;\forall\, i=1,\ldots,N \right\}\,.
\]
Clearly $0<T_R'\le +\infty$ and by continuity, in view of the assumptions \eqref{in} and \eqref{initial}, $T_R^\eps >0$ provided $\eps$ is chosen sufficiently small. Moreover, for any $t\in [0,T_R^\eps ]$, the rings evolve with supports $\Lambda_{i,\eps}(t)$ that remain separated from each other by a distance larger than or equal to $R$, and hence their mutual interaction remains bounded and Lipschitz. More precisely, during the time interval $[0,T_R^\eps]$, the $i$th vortex ring $\omega_{i,\eps}(x,t)$ evolves according to a reduced system, with external field in Eq.~\eqref{eqchar_nF} given by
\begin{equation}
\label{fk1}
F^{i,\eps}(x,t) = \sum_{j: j\ne i} \int\!\rmd y\, \tilde H(x,y)\, \omega_{j,\eps}(y,t)\;,
\end{equation}
where $\tilde H(x,y)$ is any smooth kernel such that $\tilde H (x,y) = H(x,y)$ if $|x-y|\ge R/2$. In view of the explicit form \eqref{H1}, \eqref{H2} of $H$, and the assumption \eqref{ai}, $\tilde H$ can be chosen such that $\bs F^{i,\eps} := (F^{i,\eps}_1,F^{i,\eps}_2,0)$ has zero divergence\footnote{This mollification is obtained by modifying the stream function associated to the field. The existence of such function for axisymmetric flow without swirl is a well known fact, see, e.g., \cite[Section 2]{FrB74}.} and, for some constant $D>0$, any $i,j=1,\ldots, N$, and $t\in [0,T_R^\eps]$,
\[
|F^{i,\eps}(x,t)| \le \frac{D}{|\log\eps|}\,, \quad |F^{i,\eps}(x,t) - F^{j,\eps}(y,t)| \le \frac{D}{|\log\eps|} |x-y|\quad \forall\,x,y\in\bb R^2\,.
\]
Therefore, we can apply Theorem \ref{thm:2} to the evolution of the $i$th vortex ring, with $\zeta^* = \zeta^i$, $a=a_i$, and choosing $R_*< R/2$,  to conclude that there are $T_*>0$ and $\eps_*\in (0,1)$ such that:

\noindent
(1) $\Lambda_{i,\eps}(t) \subseteq \Sigma(\zeta^i(t)|R_*)$ for any $t\in[0, T_*\wedge T_R^\eps]$, $\eps\in (0,\eps_*]$, and $i=1,\ldots,N$, where $\zeta^i(t)$ is defined in \eqref{free_m};

\noindent
(2) there are $\zeta^{i,\eps}(t)\in \Pi$, $i=1,\ldots,N$, and $\varrho_\eps>0$ such that
\[
\lim_{\eps\to 0}|\log\eps| \int_{\Sigma(\zeta^{i,\eps}(t)|\varrho_\eps)}\!\rmd x\, \omega_\eps(x,t) = a_i\,,
\]
with
\[
\lim_{\eps\to 0} \varrho_\eps = 0\,, \quad \lim_{\eps\to 0} \zeta^{i,\eps}(t) = \zeta^i(t)\,.
\]
As $R_*< R/2$, by continuity $T_* < T_R^\eps$, so that the theorem follows with $T_R=T_*$, $\eps_R=\eps_*$, and $R_\eps=\varrho_\eps$.
\qed

\section{Proof of Theorem \ref{thm:2}}
\label{sec:3}

Without loss of generality, we prove the theorem in the case $a=1$, hence
\begin{equation}
\label{MgammaFb}
0 \le \omega_\eps(x,0) \le \frac{M}{\eps^2|\log\eps|} \quad \forall\, x\in\bb R^2\,, \qquad |\log\eps|\int\!\rmd y\, \omega_\eps(y,0) =1\,.
\end{equation}

A preliminary step is a concentration result, which shows that large part of the vorticity remains confined in a disk whose size is infinitesimal as $\eps\to 0$.

\begin{lemma}
\label{lem:1}
Under the hypothesis of Theorem \ref{thm:2}, further assuming Eq.~\eqref{MgammaFb}, for each $T>0$ there are $\eps_1\in (0,1)$, $C_1>0$ and $q^\eps(t)\in\bb R^2$, $t\in [0,T]$, such that
\begin{equation}
\label{lem1a}
|\log\eps| \int_{\Sigma(q^\eps(t)|\sqrt\eps)}\!\rmd x\, \omega_\eps(x,t) \ge 1 - \frac{C_1}{|\log\eps|} \quad \forall\, t\in [0,T] \quad  \forall\, \eps\in (0,\eps_1]
\end{equation}
and
\begin{equation}
\label{lem1b}
|\log\eps| \int_{\Sigma(q^\eps(t)|\eps|\log\eps|)}\!\rmd x\, \omega_\eps(x,t) \ge 1 - \frac{C_1}{\log|\log\eps|} \quad \forall\, t\in [0,T] \quad  \forall\, \eps\in (0,\eps_1]\,.
\end{equation}
\end{lemma}

This result is a corollary of \cite[Lemma 2.1]{BCM00}, where the case without external field is considered. But this lemma remains valid also when $F^\eps \ne 0$, see Appendix \ref{app:b}.

Under the Assumption \ref{ass:1}, given $T>0$ and a radius $\bar r < r_*/2$, we define
\begin{equation}
\label{Teps}
T_\eps := \sup\{ t\in (0,T] \colon \Lambda_\eps(s) \subseteq \Sigma(\zeta(s)|\bar r) \;\;\forall\,s\in [0,t]\}\,,
\end{equation}
with $\zeta(t)$ as in \eqref{zett} (with $a=1$). In view of \eqref{initialF}, if $\eps$ is small enough then $\Lambda_\eps(0) \subset \Sigma(\zeta^*|\bar r) = \Sigma(\zeta(0)|\bar r)$, and hence $T_\eps>0$ by continuity. We next analyze the evolution during the time interval $[0,T_\eps]$.

We denote by $B^\eps(t)$ the center of vorticity of the ring, defined by
\begin{equation}
\label{c.m.}
B^\eps(t) = \frac{\int\! \rmd x\, x\, \omega_\eps(x,t)}{\int\! \rmd x\, \omega_\eps(x,t)} = |\log\eps| \int\! \rmd x\, x\, \omega_\eps(x,t)
\end{equation}
(observe that $M_0(t)=\int\! \rmd x\, \omega_\eps(x,t)$ is a constant of motion, see Eq.~\eqref{m0c} in the Appendix), and by $I_\eps(t)$ the moment of inertia with respect of $B^\eps(t)$, i.e.,
\begin{equation}
\label{moment}
I_\eps(t) = \int\! \rmd x\, |x-B^\eps(t)|^2  \omega_\eps(x,t)\,.
\end{equation}

\begin{lemma}
\label{lem:2}
Fix $T>0$ and let $\eps_1$ be as in Lemma \ref{lem:1}. Fix also $\bar r < r_*/2$ and choose $\eps_2\in (0, \eps_1]$ such that $T_\eps>0$ for any $\eps\in (0,\eps_2]$, with $T_\eps$ as in Eq.~\eqref{Teps}. Then, there exists $C_2>0$ such that
\[
I_\eps(t) \le \frac{C_2}{|\log\eps|^2}  \quad \forall\, t\in [0,T_\eps] \quad  \forall\, \eps\in (0,\eps_2]\,.
\]
\end{lemma}

\begin{proof}
We have, by \eqref{lem1a},
\[
\begin{split}
I_\eps(t) & = \min_{q\in\bb R^2} \int\! \rmd x\, |x-q|^2  \omega_\eps(x,t) \le  \int\! \rmd x\, |x-q_\eps(t)|^2  \omega_\eps(x,t) \\ & = \int_{\Sigma(q^\eps(t)|\sqrt\eps)}\! \rmd x\, |x-q^\eps(t)|^2  \omega_\eps(x,t) + \int_{\Sigma(q^\eps(t)|\sqrt\eps)^\complement}\! \rmd x\, |x-q^\eps(t)|^2  \omega_\eps(x,t) \\ & \le \frac{(\sqrt\eps)^2}{|\log\eps|} + \frac{C_1}{|\log\eps|^2}\max_{x\in\Lambda_\eps(t)} |x-q_\eps(t)|^2\,.
\end{split}
\]
Now, for $t\in [0,T_\eps]$ and $\eps$ small enough,
\[
\max_{x\in\Lambda_\eps(t)} |x-q_\eps(t)|^2 \le 2 \max_{x\in\Lambda_\eps(t)} |x|^2 + 2|q_\eps(t)|^2 \le 2(|\zeta(t)|+\bar r)^2 + 2(|\zeta(t)|+\bar r+\sqrt\eps)^2\,,
\]
where we used that, in view of Eq.~\eqref{lem1a}, $\Sigma(q^\eps(t)|\sqrt\eps)\cap\Lambda_\eps(t)\ne \emptyset$. On the other hand, since $T_\eps\le T$ we also have $|\zeta(t)| \le C(|\zeta^*|+T)$ for some $C>0$ independent of $\eps$. Then, the lemma follows from the above estimates.
\end{proof}

Our next goal is to show that the vorticity outside any small disk centered in $B^\eps$ is indeed extremely small, see Proposition \ref{prop:1} below. We first need a preliminary result which details the structure of the kernel $H(x,y)$ appearing in \eqref{u=}. As we show in Lemma \ref{lem:3} below, whose proof is postponed in Appendix \ref{app:b}, the most singular part of $H(x,y)$ is given by the kernel $K(x-y)$ corresponding to the planar case, i.e.,
\begin{equation}
\label{vel-vor3}
K(x) = \nabla^\perp  G(x)\,, \quad G(x) := - \frac{1}{2\pi} \log|x|\,,
\end{equation}
where $v^\perp :=(v_2,-v_1)$ for $v = (v_1,v_2)$.

\begin{lemma}
\label{lem:3}
There exists $C_0>0$ such that
\begin{equation}
\label{sH}
H(x,y) = K(x-y) + L(x,y) + \mc R(x,y)\,, 
\end{equation}
where $K(\cdot)$ is defined in \eqref{vel-vor3}, 
\begin{equation}
\label{sL}
L(x,y) = \frac{1}{4\pi x_2} \log\frac {1+|x-y|}{|x-y|}\begin{pmatrix} 1 \\ 0 \end{pmatrix}
\end{equation}
and, for any $x,y\in \Pi$,
\begin{equation}
\label{sR}
|\mc R(x,y)| \le C_0\frac{1+x_2+\sqrt{x_2y_2} \big(1+ |\log(x_2y_2)|\big)}{x_2^2}\,.
\end{equation}
\end{lemma}

In particular, under the hypothesis of Lemma \ref{lem:2} and recalling $\bar r<r_*/2$ we have, for a suitable $C_3>0$ and any $\eps\in (0,\eps_2]$ and $ t\in [0,T_\eps]$,
\begin{equation}
\label{lrest}
|L(x,y)| \le \frac{1}{2\pi r_*}\log\frac {1+|x-y|}{|x-y|}\,,\;\; |\mc R(x,y)| \le C_3 \quad \forall\, x,y\in \Lambda_\eps(t)\,.
\end{equation}

\begin{proposition}
\label{prop:1}
Let
\[
m_t(R) := \int_{\Sigma(B^\eps(t)|R)^\complement}\!\rmd x\, \omega_\eps(x,t)
\]
denote the amount of vorticity outside the disk $\Sigma(B^\eps(t)|R)$ at time $t$. Fix $T>0$, $\bar r < r_*/2$, and let $\eps_2$ be as in Lemma \ref{lem:2}. Then, for each $\ell>0$ there are $\bar T \in (0,T]$ and $\eps_3 \in (0,\eps_2]$ such that
\begin{equation}
\label{smt}
m_t(\bar r/8) \le \eps^{\ell} \quad \forall\, t\in [0,\bar T\wedge T_\eps] \quad \forall\,\eps\in (0,\eps_3]\,,
\end{equation}
with $T_\eps$ as in Eq.~\eqref{Teps}.
\end{proposition}

\begin{proof}
Given $R\ge 2h>0$, let $x\mapsto W_{R,h}(x)$, $x\in \bb R^2$, be a non-negative smooth function, depending only on $|x|$, such that
\begin{equation}
\label{W1}
W_{R,h}(x) = \begin{cases} 1 & \text{if $|x|\le R$}, \\ 0 & \text{if $|x|\ge R+h$}, \end{cases}
\end{equation}
and, for some $C_W>0$,
\begin{equation}
\label{W2}
|\nabla W_{R,h}(x)| < \frac{C_W}{h}\,,
\end{equation}
\begin{equation}
\label{W3}
|\nabla W_{R,h}(x)-\nabla W_{R,h}(x')| < \frac{C_W}{h^2}\,|x-x'|\,. 
\end{equation}

We define the quantity
\begin{equation}
\label{mass 1}
\mu_t(R,h) = \int\! \rmd x \, \big[1-W_{R,h}(x-B^\eps(t))\big]\, \omega_\eps (x,t)\,,
\end{equation}
which is a mollified version of $m_t$, satisfying
\begin{equation}
\label{2mass 3}
\mu_t(R,h) \le m_t(R) \le \mu_t(R-h,h)\,.
\end{equation}
In particular, it is enough to prove \eqref{smt} with $\mu_t$ instead of $m_t$. 

To this purpose, we study the time derivative of $\mu_t(h)$. By applying \eqref{weqF} with test function $f(x,t) = 1-W_{R,h}(x-B^\eps(t))$ we have, 
\[
\frac{\rmd}{\rmd t} \mu_t(R,h) = - \int\! \rmd x\, \nabla W_{R,h}(x-B^\eps(t)) \cdot [u(x,t)+ F^\eps(x,t) - \dot B^\eps(t)]\,\omega_\eps(x,t)\,.
\]
We now observe that the flow $x(0) \mapsto x(t)$ induced by \eqref{eqchar_nF} preserves the measure $x_2 \rmd x$, so that  Eq.~\eqref{cons-omr_ni} implies $\omega_\eps(x(t),t)\, \rmd x(t) = \omega_\eps(x(0),0)\, \rmd x(0)$ (see also Eq.~\eqref{m0c} in the Appendix). Therefore, from Eqs.~\eqref{c.m.} and \eqref{eqchar_nF},
\begin{equation}
\label{bpunto}
\begin{split}
\dot B^\eps(t) & = |\log\eps|\frac{\rmd}{\rmd t} \int\!\rmd x\, x\,\omega_\eps(x,t) =  |\log\eps| \int\!\rmd x\, \omega_\eps(x,t) \, (u+F^\eps)(x,t) \\ & =  |\log\eps|\int\!\rmd x\, \omega_\eps(x,t) \, F^\eps (x,t) \\ & \quad +  |\log\eps|\int\!\rmd x\, \omega_\eps(x,t) \int\!\rmd y\,  [L(x,y) + \mc R(x,y)]\omega_\eps(y,t)\,,
\end{split}
\end{equation}
where in the last equality we used Eqs.~\eqref{u=}, \eqref{sH} and that, since $K(\cdot)$ is an odd funtion,
\[
\int\!\rmd x\, \omega_\eps(x,t) \int\!\rmd y\, K(x-y) \omega_\eps(y,t) = 0\,.
\]
Applying again Eqs.~\eqref{u=} and \eqref{sH} we thus conclude that
\begin{equation}
\label{mass 4}
\frac{\rmd}{\rmd t} \mu_t(R,h) =  - (A_1 + A_2 + A_3 + A_4) \,,
\end{equation}
with
\begin{equation*}
\begin{split}
A_1 & = \int\! \rmd x\, \nabla W_{R,h}(x-B^\eps(t)) \cdot \int\!\rmd y \, K(x-y)\, \omega_\eps(y,t)\, \omega_\eps(x,t)  \\ & = \frac 12 \int\! \rmd x \! \int\! \rmd y\, [\nabla W_{R,h}(x-B^\eps(t)) - \nabla W_{R,h}(y-B^\eps(t))] \\ & \qquad  \cdot K(x-y) \, \omega_\eps(x,t)\,  \omega_\eps(y,t) \\ A_2 & = |\log\eps|  \int\! \rmd x\, \nabla W_{R,h}(x-B^\eps(t)) \cdot \int\!\rmd y \,[F^\eps(x,t)-F^\eps(y,t)]\, \omega_\eps(y,t)\, \omega_\eps(x,t)\,, \\ A_3 & = \int\! \rmd x\, \nabla W_{R,h}(x-B^\eps(t)) \cdot \int\!\rmd y \, L(x,y)\, \omega_\eps(y,t)\, \omega_\eps(x,t) \\ & \;\; - |\log\eps| \int\! \rmd x\, \nabla W_{R,h}(x-B^\eps(t)) \cdot \int\! \rmd z\, \int\! \rmd y\, L(z,y) \omega_\eps(z,t)\,  \omega_\eps(y,t)\, \omega_\eps(x,t)\,, \\ A_4 & = \int\! \rmd x\, \nabla W_{R,h}(x-B^\eps(t)) \cdot \int\!\rmd y \, \mc R(x,y)\, \omega_\eps(y,t)\, \omega_\eps(x,t) \\ & \;\; - |\log\eps| \int\! \rmd x\, \nabla W_{R,h}(x-B^\eps(t)) \cdot \int\! \rmd z\, \int\! \rmd y\, \mc R(z,y)\,  \omega_\eps(z,t)\,  \omega_\eps(y,t) \,  \omega_\eps(x,t)\,,
\end{split}
\end{equation*}
where the second expression of $A_1$ is due to the antisymmetry of $K$.

Concerning $A_1$, we introduce the new variables $x'=x-B^\eps(t)$, $y'=y-B^\eps(t)$, define $\tilde\omega_\eps(z,t) := \omega_\eps(z+B^\eps(t),t)$, and let
\[
f(x',y') = \frac 12 \tilde\omega_\eps(x',t)\, \tilde\omega_\eps(y',t) \, [\nabla W_{R,h}(x')-\nabla W_{R,h}(y')] \cdot K(x'-y') \,,
\]
so that $A_1 = \int\!\rmd x' \! \int\!\rmd y'\,f(x',y')$. We observe that $f(x',y')$ is a symmetric function of $x'$ and $y'$ and that, by \eqref{W1}, a necessary condition to be different from zero is if either $|x'|\ge R$ or $|y'|\ge R$. Therefore, 
\begin{equation*}
\begin{split}
A_1  &= \bigg[ \int_{|x'| > R}\!\rmd x' \! \int\!\rmd y' + \int\!\rmd x' \! \int_{|y'| > R}\!\rmd y' -  \int_{|x'| > h}\!\rmd x' \! \int_{|y'| > R}\!\rmd y'\bigg]f(x',y') \\ & = 2 \int_{|x'| > R}\!\rmd x' \! \int\!\rmd y'\,f(x',y')  -  \int_{|x'| > R}\!\rmd x' \! \int_{|y'| > R}\!\rmd y'\,f(x',y') \\ & = A_1' + A_1'' + A_1'''\,,
\end{split}
\end{equation*}
with 
\begin{equation*}
\begin{split}
A_1' & = 2 \int_{|x'| > R}\!\rmd x' \! \int_{|y'| \le R-h}\!\rmd y'\,f(x',y') \,, \\ A_1''&  = 2 \int_{|x'| > R}\!\rmd x' \! \int_{|y'| > R-h}\!\rmd y'\,f(x',y')\,, \\ A_1''' & = -  \int_{|x'| > R}\!\rmd x' \! \int_{|y'| > R}\!\rmd y'\,f(x',y')\,.
\end{split}
\end{equation*}
By the assumptions on $W_{R,h}$, we have $\nabla W_{R,h}(z) = \eta_h(|z|) z/|z|$ with $\eta_h(|z|) =0$ for $|z| \le R$. In particular, $\nabla W_{R,h}(y') = 0$ for $|y'| \le R-h$, hence
\[
A_1' =  \int_{|x'| > R}\!\rmd x' \, \tilde\omega_\eps(x',t) \eta_h(|x'|) \,\frac{x'}{|x'|} \cdot  \int_{|y'| \le R-h}\!\rmd y'\, K(x'-y') \, \tilde\omega_\eps(y',t)\,.
\]
In view of  \eqref{W2}, $|\eta_h(|z|)| \le C_W/h$, so that 
\begin{equation}
\label{a1'}
|A_1'| \le \frac{C_W}{h} m_t(R) \sup_{|x'| > R} |H_1(x')|\,,
\end{equation}
with
\[
H_1(x') = \frac{x'}{|x'|}\cdot  \int_{|y'| \le R-h}\!\rmd y'\, K(x'-y') \, \tilde\omega_\eps(y',t) \,.
\]
Now, recalling \eqref{vel-vor3} and using that $x'\cdot (x'-y')^\perp=-x'\cdot y'^\perp$, we get,
\begin{equation}
\label{in H_11}
H_1(x') = \frac{1}{2\pi} \int_{|y'|\leq R-h}\! \rmd y'\, \frac{x'\cdot y'^\perp}{|x'||x'-y'|^2}\,  \tilde\omega_\eps(y',t) \,.
\end{equation}
By \eqref{c.m.}, $\int\! \rmd y'\,  y'^\perp\,  \tilde\omega_\eps(y',t) = 0$, so that
\begin{equation}
\label{in H_13}
H_1(x')  = H_1'(x')-H_1''(x')\,, 
\end{equation}
where
\begin{eqnarray*}
	&& H_1'(x') = \frac{1}{2\pi}  \int_{|y'|\le R-h}\! \rmd y'\, \frac {x'\cdot y'^\perp}{|x'|}\, \frac {y'\cdot (2x'-y')}{|x'-y'|^2 \ |x'|^2} \,  \tilde\omega_\eps(y',t) \,, \\ && H_1''(x')= \frac{1}{2\pi} \int_{|y'|> R-h}\! \rmd y'\, \frac{x'\cdot y'^\perp}{|x'|^3}\,  \tilde\omega_\eps(y',t) \,.
\end{eqnarray*}
We notice that if $|x'| > R$ then $|y'| \le R-h$ implies $|x'-y'|\ge h$ and $|2x'-y'|\le |x'-y'|+|x'|$. Therefore, for any $|x'| > R \ge 2h$,
\[
\begin{split}
|H_1'(x')|& \le \frac{1}{2\pi}\bigg[\frac{1}{|x'|^2h} + \frac{1}{|x'|h^2} \bigg]  \int_{|y'|\leq R-h} \! \rmd y'\, |y'|^2 \,  \tilde\omega_\eps(y',t) \\ & \le \frac{I_\eps(t)}{2\pi}\bigg[\frac{1}{R^2h} + \frac{1}{Rh^2}\bigg] \le \frac{3I_\eps(t)}{4\pi Rh^2}\,.
\end{split}
\]
To bound $H_1''(x')$, by Chebyshev's inequality, for any $|x'| > R \ge 2h$ we have,
\[
|H_1''(x)| \le \frac{1}{2\pi |x'|^2} \int_{|y'|> R-h}\! \rmd y'\, |y'| \tilde\omega_\eps(y',t) \le \frac{I_\eps(t)}{2\pi R^2(R-h)} \le \frac{I_\eps(t)}{2 \pi R^2h} \,.
\]
From Eqs.~\eqref{a1'} and \eqref{in H_13}, the previous estimates, and $R\ge 2h$, we conclude that
\begin{equation}
\label{H_14b}
|A_1'| \le \frac{5C_WI_\eps(t)}{4\pi Rh^3} m_t(R)\,.
\end{equation}

Now, by \eqref{W3} and then applying the Chebyshev's inequality and again $R\ge 2h$,
\begin{equation*}
\begin{split}
|A_1''| + |A_1'''| & \le \frac{C_W}{\pi h^2} \int_{|x'| \ge R}\!\rmd x' \! \int_{|y'| \ge R-h}\!\rmd y'\,\tilde\omega_\eps(y',t) \,  \tilde\omega_\eps(x',t)  \\ & = \frac{C_W}{\pi h^2}m_t(R)   \int_{|y'| \ge R-h}\!\rmd y'\, \tilde\omega_\eps(y',t)  \le \frac{4C_W I_\eps(t)}{\pi R^2h^2} m_t(R)\,.
\end{split}
\end{equation*}
In conclusion, recalling $R\ge 2h$, 
\begin{equation}
\label{a1s}
|A_1| \le  \frac{13 C_W I_\eps(t)}{4\pi R h^3} m_t(R)\,.
\end{equation}

Concerning $A_2$, we observe that by \eqref{W1} the integrand is different from zero only if $R\le |x-B^\eps(t)|\le R+h$. Therefore, by item (b) in Assumption \ref{ass:1} and \eqref{W2} we have,  using again the variables $x'=x-B^\eps(t)$, $y'=y-B^\eps(t)$,
\[
\begin{split}
|A_2| & \le \frac{2C_W C_F}{h} \int_{|x'|\ge R}\!\rmd x'  \tilde\omega_\eps(x',t)  \int_{|y'|> R}\!\rmd y'\, \tilde\omega_\eps(y',t) \\ &\quad + \frac{C_WL}{h} \int_{R \le |x'|\le R+h}\!\rmd x'  \tilde\omega_\eps(x',t) \int_{|y'| \le R}\!\rmd y'\,|x'- y'| \,  \tilde\omega_\eps(y',t) \,.
\end{split}
\]
Since $|x'-y'| \le 2R+h$ in the domain on integration of the last integral and using the Chebyshev's inequality in the first one we get,
\begin{equation}
\label{a2s}
|A_2| \le \frac{2C_W C_F I_\eps(t)}{R^2h} m_t(R) +  \frac{C_WL}{|\log\eps|} \bigg(1+\frac{2R}h\bigg) m_t(R)\,.
\end{equation}

To bound $A_3$ and $A_4$, we now restrict to the case of interest, $t\in [0,T_\eps]$ and $\eps\in (0,\eps_2]$, so that the kernels $L$ and $\mc R$ can be bounded as in \eqref{lrest}. Then, using also $|\log\eps|\int\!\rmd z\, \omega_\eps(z,t)=1$,
\begin{equation}
\label{a3s}
|A_3| + |A_4| \le \frac{2C_W}h  \Big(\sup_x \alpha(x,t) + C_3 \Big) m_t(R)\,,
\end{equation}
where
\begin{equation}
\label{axt1}
\alpha(x,t) = \frac{1}{r_*} \int\!\rmd y \, \log\frac{1+|x-y|}{|x-y|}\, \omega_\eps(y,t) \,.
\end{equation}
To bound $\alpha(x,t)$ , we observe that the integrand is monotonically unbounded as $y\to x$, and so the maximum of the integral is obtained when we rearrange the vorticity mass as close as possible to the singularity. Therefore, recalling \eqref{MgammaF},
\[
\begin{split}
\alpha(x,t) & \le \frac{M}{2\eps^2|\log\eps|} \frac{1}{r_*}\int_0^{\bar\rho}\!\rmd \rho\, \rho \, \log\frac{1+\rho}{\rho} \\ & =  \frac{M}{2\eps^2|\log\eps|r_*}  \bigg\{\frac{\bar\rho^2}{2} \log\frac{1+\bar\rho}{\bar\rho} - \frac 12 \int_0^{\bar\rho}\!\rmd \rho\, \frac{\rho}{1+\rho} \bigg\}\,,
\end{split}
\]
with $\bar\rho$ such that $\pi\bar\rho^2 M/(\eps^2|\log\eps|)= 1/|\log\eps|$, and hence, for some $C_4>0$
\begin{equation}
\label{axt}
\sup_x \alpha(x,t) \le C_4\,.
\end{equation}
From \eqref{a1s}, \eqref{a2s}, \eqref{a3s}, and Lemma \ref{lem:2} we deduce that
\begin{equation}
\label{2mass 4''}
\frac{\rmd}{\rmd t} \mu_t(R,h) \le A_\eps(R,h) m_t(h)\quad \forall\, t\in [0,T_\eps] \quad  \forall\, \eps\in (0,\eps_2]\,,
\end{equation}
where, for some $C_5>0$ and any $R>2h$,
\begin{equation}
\label{mass 4bis}
A_\eps(R,h) = \frac{C_5}{h} \bigg( \frac{1}{|\log\eps|^2Rh^2} + \frac{1}{|\log\eps|^2R^2} + \frac{R}{|\log\eps|} + 1 \bigg)\,.
\end{equation}
Therefore, by \eqref{2mass 3} and \eqref{2mass 4''},
\begin{equation}
\label{mass 14'}
\mu_t(R,h) \le \mu_0(h) + A_\eps(R,h) \int_{0}^t \rmd s\, \mu_s(R-h,h) \quad \forall\, t\in [0,T_\eps] \quad  \forall\, \eps\in (0,\eps_2]\,.
\end{equation}
We iterate the last inequality $n = \lfloor|\log\eps|\rfloor$ times,\footnote{$\lfloor a\rfloor$ denotes the integer part of the positive number $a$.} from $R_0=\bar r/8-h$ to $R_n = \bar r/8 -(n+1)h = \bar r/16$. Since $h = \bar r/(16n) $ and $R_j\in [\bar r/16, \bar r/8]$, from the explicit expression \eqref{mass 4bis} it is readily seen that there is $A_*>0$ such that $A_\eps(R_j,h) \le A_*n/\bar r$ for any $j=0,\ldots,n$ and $\eps\in (0,\eps_2]$. Therefore, for any $\eps\in (0,\eps_2]$ and $t\in [0,T_\eps]$, 
\[
\begin{split}
\mu_t(\bar r/8-h,h) & \le \mu_0(\bar r/8-h,h) + \sum_{j=1}^n \mu_0(R_j,h) \frac{(A_*nt/\bar r)^j}{j!} \\ & \quad + \frac{(A_*n/\bar r)^{n+1}}{n!} \int_0^t\!\rmd s\,  (t-s)^n\mu_s(R_{n+1},h) \,.
\end{split}
\]
Since $\Lambda_\eps(0) \subset \Sigma(\zeta^*|\eps)$, we can determine $\eps_3\in (0,\eps_2]$ such that $\mu_0(R_j,h)=0$ for any $j=0,\ldots,n$, so that, for any $t\in [0,T_\eps]$ and $\eps\in (0,\eps_3]$,
\begin{equation}
\label{mass 15'}
\mu_t(\bar r/8-h,h) \le \frac{(A_*n/\bar r)^{n+1}}{n!} \int_0^t\!\rmd s\,  (t-s)^n\mu_s(R_{n+1},h) \le  \frac{(A_*nt/\bar r)^{n+1}}{|\log\eps| (n+1)!}\,,
\end{equation}
where the obvious estimate $\mu_s(R_{n+1},h) \le |\log\eps|^{-1}$ has been used in the last inequality. In conclusion, using also \eqref{2mass 3}, for suitable constants $C',C''>0$,
\[
m_t(\bar r/8) \le \mu_t(\bar r/8 -h,h) \le C' \bigg(\frac{C''t}{\bar r}\bigg)^{\lfloor|\log\eps|\rfloor} \;\; \forall\, t\in [0,T_\eps] \quad  \forall\, \eps\in (0,\eps_2]\,,
\]
which implies the bound \eqref{smt} for a suitable choice of $\bar T$ and $\eps_3$.
\end{proof}

To show that the support of vorticity remains bounded as $\eps\to 0$, we need to evaluate the force acting on the fluid particles furthest from the center of vorticity.

\begin{lemma}
\label{lem:4}
Fix $T>0$, $\bar r < r_*/2$, and let $T_\eps$ be as in definition \eqref{Teps} and $\eps_2$ be as in Lemma \ref{lem:2}. Recall $\Lambda_\eps(t)=\supp\omega_\eps(\cdot,t)$ and define
\begin{equation}
\label{Rt}
R_t:= \max\{|x-B^\eps(t)|\colon x\in \Lambda_\eps(t)\}\,.
\end{equation}
Given $\eps\in (0,\eps_2]$ and $x_0\in\Lambda_\eps(0)$, let $x(t)$ be the solution to \eqref{eqchar_nF} with initial condition $x(0) = x_0$ and suppose at time $t\in (0,T_\eps)$ it happens that 
\begin{equation}
\label{hstimv}
|x(x_0,t)-B^\eps(t)| = R_t\,.
\end{equation}
Then, at this time $t$, 
\begin{equation}
\label{stimv}
\frac{\rmd}{\rmd t} |x(t)- B^\eps(t)| \le 2L R_t + \frac{5C_2}{\pi|\log\eps|^2 R_t^3} + C_6 + \sqrt{\frac{M m_t(R_t/2)}{\pi\eps^2|\log\eps|}}\,,
\end{equation}
with $M$ and $L$ as in Assumption \ref{ass:1}, $C_2$ as in Lemma \ref{lem:2}, and $C_6 := 2(C_3+C_4)$, with $C_3$ and $C_4$ given in Eq.s~\eqref{lrest} and \eqref{axt} respectively.
\end{lemma}

\begin{proof}
Letting $x=x(t)$, from \eqref{u=}, \eqref{MgammaFb}, \eqref{sH}, and \eqref{bpunto} we have,
\begin{equation}
\label{distance1}
\begin{split}
& \frac{\rmd}{\rmd t} |x(t)- B^\eps(t)| = \big(u(x,t) + F^\eps(x,t) - \dot B^\eps(t)\big) \cdot \frac{x-B^\eps(t)}{|x-B^\eps(t)|} \\ & \;\; = \int\!\rmd y\, \omega_\eps(y,t) \,  \big[|\log\eps|\big(F(x,t) - F(y,t) \big) + K(x-y)\big] \cdot  \frac{x-B^\eps(t)}{|x-B^\eps(t)|} \\ & \quad + \int\!\rmd y\,  \omega_\eps(y,t) \, \big[L(x,y)+\mc R(x,y)\big] \cdot  \frac{x-B^\eps(t)}{|x-B^\eps(t)|}  \\ & \quad - |\log\eps|\int\!\rmd y\, \omega_\eps(y,t) \int\!\rmd z\,  \omega_\eps(z,t) \big[L(y,z)+\mc R(y,z)\big] \cdot  \frac{x-B^\eps(t)}{|x-B^\eps(t)|}
\end{split}
\end{equation}
(recall $\int\! \rmd x\, \omega_\eps(x,t) = |\log\eps|^{-1}$ for any $t\ge 0$).

The integral terms in the last two lines have been already estimated to get Eq.~ \eqref{a3s}. In view of \eqref{axt}, we deduce that the sum of the last two terms in the right-hand side of Eq.~\eqref{distance1} is bounded by $2(C_3+C_4) =:C_6$.

The first term in the second line, due the external field, is easily bounded by hypothesis (b) of Assumption \ref{ass:1} and \eqref{Rt},
\begin{equation}
\label{distance4}
\bigg||\log\eps|\int\!\rmd y\, \omega_\eps(y,t) \, [F(x,t) - F(y,t)] \bigg| \le L|\log\eps| \int\!\rmd y\,  \omega_\eps(y,t) \,|x-y|  \le 2L R_t\,.
\end{equation}

For the second term, we split the domain of integration in two parts, the disk $\mc D=\Sigma(B^\eps(t)|R_t/2)$ and the annulus $\mc A= \Sigma(B^\eps(t)|R_t)\setminus\Sigma(B^\eps(t)|R_t/2)$. Then,
\begin{equation}
\label{in A_1,A_2}
\frac{x-B^\eps(t)}{|x-B^\eps(t)|} \cdot \int\! \rmd y\, K(x-y)\, \omega_\eps(y,t) = H_\mc D + H_\mc A\,,
\end{equation}
where
\begin{equation}
\label{in A_1}
H_\mc D(x) = \frac{x-B^\eps(t)}{|x-B^\eps(t)|} \cdot \int_{\mc D}\! \rmd y\, K(x-y)\, \omega_\eps(y,t) 
\end{equation}
and
\begin{equation}
\label{in A_2}
H_\mc A(x)= \frac{x-B^\eps(t)}{|x-B^\eps(t)|} \cdot \int_{\mc A}\! \rmd y\, K(x-y)\, \omega_\eps(y,t)\,.
\end{equation}

We first evaluate the contribution of the integration on $\mc D$. We notice that $H_\mc D(x)$ is exactly equal to the integral $H_1(x')$ appearing in Eq.~\eqref{a1'}, provided $x'=x-B^\eps(t)$ and $R=2h=R_t$. Moreover, to obtain Eq.~\eqref{H_14b} we had to bound $H_1(x')$ for $|x'|\ge R$, which is exactly what we need now, as $|x-B^\eps(t)|=R_t$.  This estimate, adapted to the present context becomes
\begin{equation}
\label{H_14}
|H_\mc D| \le \frac{5 I_\eps(t)}{\pi R_t^3} \le \frac{5C_2}{\pi|\log\eps|^2 R_t^3}\,,
\end{equation}
where we applied Lemma \ref{lem:2} in the last inequality.

We now evaluate $H_\mc A$. Recalling the definition of $K$, 
\begin{equation*}
|H_\mc A| \le \frac{1}{2\pi} \int_{A_2}\! \rmd y\, \frac 1{|x-y|} \, \omega_\eps(y,t)\,.
\end{equation*}
Since the integrand is monotonically unbounded as $y\to x$, we can argue as done to get Eq.~\eqref{axt}: the maximum possible value of the integral is obtained when we rearrange the vorticity mass as close as possible to the singularity. In view of the assumption \eqref{MgammaFb} and since $m_t(R_t/2)$ is equal  to the total amount of vorticity in $\mc A$, this rearrangement reads,
\begin{equation}
\label{h2}
|H_\mc A| \le \frac{M\eps^{-2}}{2\pi} \int_{\Sigma (0|\bar \rho)}\!\rmd y'\, \frac{1}{|y'|} = M\eps^{-2} \bar \rho\,, 
\end{equation}
where the radius $\bar\rho$ is such that $\pi\bar\rho^2 M/(\eps^2|\log\eps|) = m_t(R_t/2)$. The estimate \eqref{stimv} now follows by \eqref{distance1}, \eqref{distance4}, \eqref{in A_1,A_2}, \eqref{H_14}, and \eqref{h2}.
\end{proof}

\begin{lemma}
\label{cor:1}
Fix $T>0$, $\bar r < r_*/2$, and let $T_\eps$ be as in definition \eqref{Teps}, and $\bar T \in (0,T]$ and $\eps_3 \in (0,1)$ be as in Proposition \ref{prop:1} with $\ell=2$. Then, there exists $\bar T_1\in (0,\bar T]$ such that 
\begin{equation}
\label{supp_pro}
\Lambda_\eps(t) \subset \Sigma(B^\eps(t)|\bar r/2) \quad \forall\, t\in [0,\bar T_1\wedge T_\eps] \quad \forall\,\eps\in (0,\eps_3]\,.
\end{equation}
\end{lemma}

\begin{proof}
Fix $T>0$ and let $T_\eps$ be as in definition \eqref{Teps} and $\eps_2$ be as in Lemma \ref{lem:2}. From Lemma \ref{lem:4} it follows that if $\eps\in (0,\eps_2]$ then $\Lambda_\eps(t) \subset \Sigma(B^\eps(t)|R(t))$ for any $t\in (0,T_\eps)$, where $R(t)$ is a solution to 
\begin{equation}
\label{stimrbis}
\dot R(t) = 2L R(t) + \frac{5C_2}{\pi|\log\eps|^2 R(t)^3} + C_6 + \sqrt{\frac{M m_t(R(t)/2)}{\pi\eps^2|\log\eps|}}\,, \quad R(0) = \eps\,.
\end{equation}
Indeed, this is true for $t=0$ and, if at some time $t\in (0,T_\eps)$ a fluid particle initially located at $x(0) = x_0\in \Lambda_\eps(0)$ reaches the boundary of $\Sigma(B^\eps(t)|R(t))$, then $R(t) = |x(t)-B^\eps(t)| =  R_t $ necessarily and hence, by \eqref{stimv}, the radial velocity of $x(t)- B^\eps(t)$ is less than or equal to $\dot R(t)$. 

Let now $\bar T\in (0,T]$ and $\eps_3 \in (0,\eps_2]$ be as in Proposition \ref{prop:1} with $\ell=2$. For $\eps\in (0,\eps_3]$ we define
\[
t_1 := \inf\{ t\in (0,\bar T\wedge T_\eps] \colon R(t) = \bar r/2\}\,,
\]
setting $t_1 = \bar T\wedge T_\eps$ if $R(t)<\bar r/2$ for any $t\in [0,\bar T]$. Since $\Lambda_\eps(t) \subset \Sigma(B^\eps(t)|R(t))$ for any $t\in (0,T_\eps)$, the claim of the lemma reduces to exhibit $\bar T_1\in (0,\bar T]$ such that $t_1 \ge \bar T_1\wedge T_\eps$ for any $\eps\in (0,\eps_3]$. 

If $t_1= \bar T\wedge T_\eps$ the claim follows with $\bar T_1 = \bar T$. Otherwise, if $t_1<\bar T\wedge T_\eps$ we define
\[
t_0 = \inf\{t\in [0,t_1]\colon R(s) > \bar r/4 \;\forall\, s\in [t,t_1] \}
\]
and notice that $R(t_1) = \bar r/2$, $R(t_0) = \bar r/4$, and $R(t) \in [\bar r/4,\bar r/2]$ for any $t\in [t_0,t_1]$. In particular,  $m_t(R(t)/2)\le m(\bar r/8) \le \eps^2$ for any $t\in [t_0,t_1]$. Therefore, by \eqref{stimrbis},
\begin{equation*}
\dot R(t) \le 2L \bar r + \frac{40C_2}{\pi|\log\eps|^2 \bar r^3} + C_6 + \sqrt{\frac{M }{\pi|\log\eps|}} \quad \forall\, t\in [t_0,t_1]\,.
\end{equation*}
This means that there is $C_7>0$ such that $\dot R(t) \le C_7$ for any $ t\in [t_0,t_1]$ and $\eps\in (0,\eps_3]$, which implies the claim by choosing $\bar T_1 < \bar r/(4C_7)$.
\end{proof}

\begin{proof}[Proof of Theorem \ref{thm:2}]
Let $R_*\in (0,r_*)$ be given. We fix $T>0$, $\bar r < R_*/2$, and let $0<\eps_3 \le \eps_2$, $\bar T$, and $\bar T_1$ be as before. 

From Lemma \ref{lem:2} we have, for any $t\in[0,\bar T_1\wedge T_\eps]$ and $\eps\in (0,\eps_3]$,
\begin{equation}
\label{Bq}
|B^\eps(t) - q^\eps(t)| = \bigg| |\log\eps| \int\!\rmd x\, \omega_\eps(x,t) |x-q^\eps(t)| \bigg| \le \sqrt{\frac{C_2}{|\log\eps|}}\,.
\end{equation}
We next prove that
\begin{equation}
\label{bz}
\lim_{\eps\to 0} \max_{t\in [0,\bar T_1\wedge T_\eps]} |B^\eps(t) - \zeta(t)| = 0\,,
\end{equation}
with $\zeta(t)$ as in \eqref{zett}  (with $a=1$). From the definition \eqref{Teps} of $T_\eps$ and \eqref{supp_pro} this implies that $T_\eps > \bar T_1$ for any $\eps$ small enough. Therefore, also in view of Eqs.~\eqref{lem1a} and \eqref{Bq}, the theorem follows with $T_* = \bar T_1$, $\eps_*\in (0,\eps_3]$ small enough, $\zeta^\eps(t) = q^\eps(t)$, and $\rho_\eps =\sqrt{\eps}$.

From Eq.~\eqref{bpunto}, item (b) in Assumption \ref{ass:1}, Eq.~\eqref{lrest}, and the explicit expression of $L(x,y)$ given in \eqref{sL} we have,
\begin{equation}
\label{b12}
|\dot B^\eps_1(t) - Q_\eps(t)| + |\dot B^\eps_2(t)| \le 2 \frac{C_F+C_3}{|\log\eps|}\quad \forall\, t\in [0,\bar T_1\wedge T_\eps] \quad \forall\,\eps\in (0,\eps_3]\,,
\end{equation}
where
\[
Q_\eps(t) := |\log\eps|\int\!\rmd x\, \omega_\eps(x,t) \frac{1}{4\pi x_2} \int\!\rmd y\, \log\frac {1+|x-y|}{|x-y|} \omega_\eps(y,t) \,.
\]
In particular, recalling also \eqref{initialF}, we deduce that 
\begin{equation}
\label{q1}
\lim_{\eps\to 0} \max_{t\in [0,\bar T_1\wedge T_\eps]} |B^\eps_2(t) - r_*| = 0 \,.
\end{equation}
To compute the asymptotic as $\eps\to 0$ of $Q_\eps(t)$, $t\in [0,\bar T_1\wedge T_\eps]$, we decompose,
\[
Q_\eps(t) := Q_\eps^1(t) + Q_\eps^2(t)
\]
with 
\[
\begin{split}
Q_\eps^1(t) & :=  |\log\eps| \int_{\Sigma(q^\eps(t),\eps|\log\eps|)}\!\rmd x\, \omega_\eps(x,t) \\ & \qquad \times  \frac{1}{4\pi x_2} \int_{\Sigma(q^\eps(t),\eps|\log\eps|)}\!\rmd y\, \log\frac {1+|x-y|}{|x-y|} \omega_\eps(y,t)\,.
\end{split}
\]
The rest $Q_\eps^2(t) = Q_\eps(t) - Q_\eps^1(t)$ is the sum of three terms, each one is the integration of the same function, which in view of Eq.~\eqref{lrest} is bounded by 
\[
\mc G(x,y) := \frac{1}{2\pi r_*} \log\frac {1+|x-y|}{|x-y|} \omega_\eps(x,t)\omega_\eps(y,t)\,,
\]
on a region where at least one between the $x$ and the $y$ variable is confined to the set $\Sigma(q^\eps(t), \eps|\log\eps|)^\complement$. Therefore, since $\mc G$ is symmetric, 
\[
\begin{split}
Q_\eps^2(t) & \le \frac{3|\log\eps|}{2\pi r_*} \int_{\Sigma(q^\eps(t),\eps|\log\eps|)^\complement}\!\rmd x\, \omega_\eps(x,t) \int\!\rmd y\, \log\frac {1+|x-y|}{|x-y|} \omega_\eps(y,t) \\ & \le  \frac{C_1}{\log|\log\eps|} \sup_x \alpha(x,t) \le \frac{C_1C_4}{\log|\log\eps|}\,,
\end{split}
\]
where $\alpha(x,t)$ is defined in \eqref{axt1} and bounded in \eqref{axt}, and we used Eq.~\eqref{lem1b}.

Concerning $Q_\eps^1(t)$, inserting a lower bound to $\frac{1}{4\pi x_2}\log\frac {1+|x-y|}{|x-y|}$ in the domain of integration and then using Eq.~\eqref{lem1b} we obtain,
\begin{equation}
\label{q2}
\begin{split}
Q_\eps^1(t) & \ge \frac{|\log\eps|}{4\pi (q^\eps_2(t)+\eps|\log\eps|)} \log\frac {1+2\eps|\log\eps|}{2\eps|\log\eps|} \bigg(\int_{\Sigma(q^\eps(t),\eps|\log\eps|)}\!\rmd x\, \omega_\eps(x,t)\bigg)^2 \\ & \ge  \frac{|\log\eps|}{4\pi (q^\eps_2(t)+\eps|\log\eps|)} \log\frac {1+2\eps|\log\eps|}{2\eps|\log\eps|} \frac{1}{|\log \eps|^2}\bigg( 1 - \frac{C_1}{\log|\log\eps|}\bigg)^2\,.
\end{split}
\end{equation}
On the other hand, by the same argument leading to \eqref{axt},
\begin{equation}
\label{q3}
\begin{split}
Q_\eps^1(t) & \le \frac{1}{4\pi (q^\eps_2(t)-\eps|\log\eps|)} \sup_x \int\!\rmd y\, \log\frac {1+|x-y|}{|x-y|} \omega_\eps(y,t) \\ & \le  \frac{ M}{2\eps^2 |\log\eps|(q^\eps_2(t)-\eps|\log\eps|)}  \bigg\{\frac{\bar\rho^2}{2} \log\frac{1+\bar\rho}{\bar\rho} - \frac 12 \int_0^{\bar\rho}\!\rmd \rho\, \frac{\rho}{1+\rho} \bigg\}\,, 
\end{split}
\end{equation}
with $\bar\rho$ such that $\pi\bar\rho^2 M/(\eps^2|\log\eps|)= 1/|\log\eps|$. From \eqref{Bq} and \eqref{q1} we deduce that the right-hand side in both \eqref{q2} and \eqref{q3} converges to $1/(4\pi r_*)$ as $\eps\to 0$, so that, in view of \eqref{b12},
\[
\lim_{\eps\to 0} \max_{t\in [0,\bar T_1\wedge T_\eps]} \bigg|B^\eps_1(t) - \bigg(z_* +\frac{t}{4\pi r_*}\bigg) \bigg| = 0 \,,
\]
which, together with Eq.~\eqref{q1}, proves \eqref{bz}.
\end{proof}

\section{Related problems}
\label{sec:4}

First of all, we recall the so-called vortex-wave system, a quite natural generalization of the planar case, which has been introduced in \cite{MaP91}. In this system, a smooth component of vorticity and point vortices evolve together via the Euler equation. In presence of axisymmetry without swirl, the methods of the present paper should allow to describe a similar mixed system, where vortex rings evolve in a background of smooth vorticity. In the rest of the section we discuss two open problems which are closely related to the subject of this paper.

\subsection{Large vortex rings}
\label{sec:4.1}

We suppose the initial configuration is given by $N$ vortex rings with a large radius, i.e., supported in disks $\Sigma((z_i,r_\eps+r_i),\eps)$ with $r_\eps$ large, and we study their evolution in the double limit when $\eps \to 0$ and $r_\eps \to \infty$ simultaneously. In this case, the intensities $A_i := \int\! \rmd x \, \omega_{i,\eps}(x,0)$ of each vortex ring are kept fixed independent of $\eps$.

We firstly consider the case $r_\eps = \const \eps^{-\alpha}$, where $\alpha >0$. In the limit $\eps \to 0$, the system converges to the point vortex model (briefly discussed in Appendix \ref{app:a}). This has been firstly proven in \cite {Mar99} for any fixed time and, more recently, extended  in \cite{CS} to time intervals diverging as $\eps \to 0$. If instead $r_0 = \const |\log\eps|$ then, at least formally, the system converges to a dynamical system, similar to the point vortex model, introduced and studied in \cite {MN}. It is defined by the following system of ordinary differential equations,
\[
\dot{z}_i(t) = \sum_{\substack{j=1 \\ j\ne i}}^N A_j K(z_i(t)-z_j (t)) +  A_i \begin{pmatrix} 1 \\ 0 \end{pmatrix}\,,
\]
with $K(\cdot)$ as in Eq.~\eqref{vel-vor3}. In the case of a vortex ring alone, when the motion reduces to a uniform translation, the convergence at any positive time has been proven rigorously in \cite{MN}, while for $N>1$ the problem is open and it does not seem easy to solve.

\subsection{Vanishing viscosity limit}
\label{sec:4.2}

In presence of a viscosity $\nu$ the incompressible fluid is governed by the Navier-Stokes equation. Moreover, in absence of a boundary, smooth solutions of the Navier-Stokes equation converge to the smooth solutions of the Euler equations as $\nu\to 0$. The case in which the initial data become singular is much more complicate.

In the case of planar symmetry the problem has been solved for vortices with intensities of the same sign \cite {Mar90}, when $\nu\le \const \eps^{- \alpha}$ with any $\alpha>0$ \cite{Mar98}, and when $\nu \to 0$ independently of $\eps \to 0$ \cite{Gal11}.

In the case of axisymmetry without swirl the results are less satisfactory. For $r_0$ fixed and $N=1$, as far as we know, the only result is for $\nu \le \const \eps^2 |\log\eps|^\alpha$ with $\alpha < 1$ \cite{BrM}. However, it is reasonable that a similar result could be extended to the case of a vortex ring of radius $r_0= \const |\log\eps|$. When $r_0 = \const \eps^{-\alpha}$, $\alpha >0$, the vanishing viscosity limit has been proven for a vortex alone at any fixed time in \cite{Mar07}, but the extension to the case of $N$ vortex rings and very long times (i.e., diverging as $\eps\to 0$) should not require too much effort. 

\appendix 

\section{Planar symmetry}
\label{app:a}

In this appendix we briefly recall some results concerning the time evolution of concentrated Euler flows with planar symmetry (without giving a complete list of references on this topic). The Euler equations for an incompressible inviscid fluid in the whole space with planar symmetry and constant density reads
\[
\partial_t\omega(x,t) + (u \cdot \nabla) \omega(x,t,) = 0\,, \quad \nabla \cdot u(x,t)  = 0\,,\quad x\in \bb R^2\,.
\]
By assuming that $u$ vanishes at infinity, the velocity field is reconstructed from the vorticity as
\[
u(x,t) = \int\!\rmd y\, K(x-y) \, \omega(y,t)\,, 
\]
with $K(\cdot)$ as in Eq.~\eqref{vel-vor3}.

We assume that initially the vorticity is concentrated in $N$ blobs of the form
\[
\omega_\eps(x,0) = \sum_{i=1}^N \omega_{i,\eps}(x,0)\,,
\]
where $\omega_{i,\eps}(x,0)$, $i=1,\ldots, N$, are functions with definite sign such that
\[
\Lambda_{i,\eps}(0) := \supp\, \omega_{i,\eps}(\cdot,0) \subset \Sigma(z_i|\eps)\,, \quad  \Sigma(z_i|\eps)\cap \Sigma(z_j|\eps)=\emptyset\quad \forall\, i \ne j\,, 
\]
with $\eps \in (0,1)$ a small parameter and the points $z_i\in \bb R^2$.

In this case, the solution of the Euler equations is strictly related to the \textit{point vortex model}, the dynamical system defined by the following system of ordinary differential equations,  
\begin{equation}
\label{Pointv}
\dot{z}_i(t) = \sum_{\substack{j=1 \\ j\ne i}}^N A_j K(z_i(t)-z_j (t))\,,
\end{equation}
where $A_j$ is called the ``intensity'' of the $j$th vortex. More precisely, it has been proven that in general, for small $\eps$, the time evolution of the vorticity has the same form,
\[
\Lambda_{i,\eps}(t)  := \supp\, \omega_{i,\eps}(\cdot,r_\eps(t)) \subset \Sigma(z_i(t),r_\eps(t))\,,
\]
with
\[
\Sigma(z_i(t),r_\eps(t)) \cap  \Sigma(z_j(t),r_\eps(t)) =0\quad \forall\, i \ne j\;, 
\]
where $r_\eps(t)$ is a positive function and $\{z_i(t); i=1,\ldots, N\}$ is the solution to Eq.~\eqref{Pointv} with initial conditions $z_i(0)=z_i$ and intensity $A_i = \int\! \rmd x \, \omega_{i,\eps}(x,0)$. The point $z_i(t)$ in the plane thus identifies a straight line in the space around which the vorticity is concentrated. 

When all the intensities $A_i$ have the same sign (positive or negative) Eq.~\eqref{Pointv} has a global in time solution. Instead, if the signs are different there are examples in which collapses could occur (e.g., two vortices collide or a vortex goes to infinity in finite time). However, these pathological events are exceptional. For a review on this issue see, for instance, \cite {MaP94}.

The point vortex model has been introduced in the eighteenth century by Helm\-holtz as a ``solution" of the Euler equations, and widely analyzed in many papers. One hundred years later, it has been considered as a numerical approximation of the Euler evolution for very irregular initial data. As a numerical tool, this system is considered when $N$ is very large and the intensity of each vortex very small (of the order of $N^{-1}$).  On this topic there are several papers, we only quote the recent review \cite{CGP14}.

As explained at the beginning of this appendix, a different point of view is adopted for finite $N$ and consists in considering the point vortex model as an approximation of $N$ very concentrated vortices, say with support of diameter 2$\eps \to 0$. It is worthwhile to emphasize that it cannot be an approximation of each evolved path, because the length of the trajectory of a fluid element diverges as $\eps \to 0$. On the other hand, by virtue of rapid rotations, a system of $N$ disjoint concentrated patches of vorticity converges as a measure to a linear combination of Dirac measures $\sum_{i=1}^N A_i \delta_{z_i(t)}$ for any positive time. This convergence has been proven 25 years ago. Recently, the problem on how long the sharp localization of the vorticity remains valid has been analyzed in \cite{BuM}.

It is possible to introduce a small viscosity $\nu$ and study the small viscosity limit $\nu \to 0$. For finite times  the status of art has been already discussed in Section \ref{sec:4.2}. The validity of the convergence on very long times (diverging as $\eps\to 0$) is analyzed in \cite {CS}.

\section{Proof of some technical results}
\label{app:b}

\begin{proof}[Proof of Lemma \ref{lem:3}]
Letting 
\begin{equation}
\label{ar}
a(x,y) := \frac{|x-y|}{\sqrt{x_2y_2}}\,,
\end{equation}
we have,
\[
H(x,y) = - \frac{I_1(a(x,y))}{2\pi} \frac{(x-y)^\perp }{x_2\sqrt{x_2y_2}} + \frac{I_2(a(x,y))}{2\pi}\frac{1}{x_2}\sqrt{\frac{y_2}{x_2}} \begin{pmatrix} 1 \\ 0 \end{pmatrix}\,,
\]
where
\begin{equation}
\label{I1I2}
I_1(a) = \int_0^\pi\!\rmd\theta\, \frac{\cos\theta}{[a^2+2(1-\cos\theta)]^{3/2}} \,, \quad I_2(a) = \int_0^\pi\!\rmd\theta\, \frac{1-\cos\theta}{[a^2+2(1-\cos\theta)]^{3/2}}\,.
\end{equation}
By an explicit computation, see, e.g., the Appendix in \cite{Mar99}, for any $a>0$,
\[
I_1(a) = \frac{1}{a^2} + \frac 14 \log\frac{a}{1+a} + \frac{c_1(a)}{1+a}\,, \quad I_2(a) = -\frac 12 \log\frac{a}{1+a} + \frac{c_2(a)}{1+a}\,,
\]
with $c_1(a)$, $c_1'(a)$, $c_2(a)$, $c_2'(a)$ uniformly bounded for $a\in (0,+\infty)$. Therefore, the  kernel $R(x,y)$ defined by \eqref{sH} is given by
\[
\mc R (x,y) = \sum_{j=1}^6 R^j(x,y)\,,
\]
with, for $a = a(x,y)$ as in \eqref{ar},
\begin{align*}
R^1(x,y) & = \frac{1}{2\pi} \bigg(1 - \sqrt{\frac{y_2}{x_2}}\bigg) \frac{(x-y)^\perp }{|x-y|^2}\,,\;\; R^2(x,y) =  \frac{1}{8\pi} \bigg(\log\frac{1+a}{a}\bigg) \frac{(x-y)^\perp }{x_2\sqrt{x_2y_2}}\,, \\ R^3(x,y) & = \frac{1}{4\pi x_2} \sqrt{\frac{y_2}{x_2}} \bigg(\log\frac{|x-y|}{1+|x-y|} - \log\frac{a}{1+a}\bigg) \begin{pmatrix} 1 \\ 0 \end{pmatrix} \,, \\ R^4(x,y) & = \frac{1}{4\pi x_2} \bigg(1-\sqrt{\frac{y_2}{x_2}}\bigg) \log\frac{|x-y|}{1+|x-y|} \begin{pmatrix} 1 \\ 0 \end{pmatrix} \,, \\ R^5(x,y)  & = -\frac{c_1(a)}{2\pi(1+a)} \frac{(x-y)^\perp }{x_2\sqrt{x_2y_2}}\,, \quad R^6(x,y) = \frac{c_2(a)}{2\pi(1+a)x_2}\sqrt{\frac{y_2}{x_2}} \begin{pmatrix} 1 \\ 0 \end{pmatrix}\,.
\end{align*}
Using that
\[
\bigg|1-\sqrt{\frac{y_2}{x_2}} \bigg| = \frac{|y_2-x_2|}{x_2+\sqrt{x_2y_2}} \le \frac{|x-y|}{x_2}
\]
and
\[
\bigg|\log\frac{|x-y|}{1+|x-y|} - \log\frac{a}{1+a}\bigg| = \bigg| \log\frac{1+a}{(x_2y_2)^{-1/2}+a}\bigg| \le \frac 12 |\log(x_2y_2)|\,,
\]
we have,
\begin{align*}
& |R^1(x,y)| = \frac{1}{2\pi} \bigg|1-\sqrt{\frac{y_2}{x_2}} \bigg| \frac{1}{|x-y|}\le \frac{1}{2\pi x_2}\,, \\
& |R^2(x,y)| = \frac{1}{8\pi x_2} \bigg(\log\frac{1+a}{a}\bigg) \frac{|x-y|}{\sqrt{x_2y_2}} \le \frac{1}{8\pi x_2} \, \sup_{\rho>0}\bigg( \rho\log\frac{1+\rho}{\rho}\bigg)\,, \\
& |R^3(x,y)| + |R^6(x,y)| \le  \frac{1}{4\pi x_2} \sqrt{\frac{y_2}{x_2}} \bigg(|\log(x_2y_2)| + \sup_{\rho>0} \frac{2c_2(\rho)}{1+\rho}\bigg)\,, \\
& |R^4(x,y)| = \frac{1}{4\pi x_2} \bigg|1-\sqrt{\frac{y_2}{x_2}} \bigg| \log\frac{1+|x-y|}{|x-y|} \le \frac{1}{4\pi x_2^2} \,\sup_{\rho>0}\bigg( \rho\log\frac{1+\rho}{\rho}\bigg)\,, \\
& |R^5(x,y)| = \frac{|c_1(a)|}{2\pi (1+a)} \frac{|x-y|}{x_2\sqrt{x_2y_2}} \le \frac{1}{2\pi x_2} \, \sup_{\rho>0}\frac{\rho c_1(\rho)}{1+\rho}\,.
\end{align*}
In conclusion, there is $C_0>0$ such that
\[
\begin{split}
& |R^1(x,y)| + |R^2(x,y)| + |R^5(x,y)| \le \frac{C_0}{x_2}\,, \quad |R^4(x,y)| \le \frac{C_0}{x_2^2}\,, \\ & |R^3(x,y)| + |R^6(x,y)| \le  \frac{C_0}{x_2} \sqrt{\frac{y_2}{x_2}} \bigg(1+|\log(x_2y_2)|\bigg)\,.
\end{split}
\]
The lemma is thus proven.
\end{proof}

\begin{proof}[Proof of Lemma \ref{lem:1}]
The proof of \cite[Lemma 2.1]{BCM00} is based on the computation of the asymptotic behavior as $\eps\to 0$ of the kinetic energy
\[
E =  \frac 12 \int\!\rmd\bs\xi\, |\bs u (\bs\xi,t)|^2 = \frac 12 \int\! \rmd x \,  2\pi x_2 |u(x,t)|^2\,,
\]
which follows from the assumptions on the initial vorticity, Eq.~\eqref{cons-omr}, and the conservation along the motion of $E$ and of the quantities
\[ 
M_0  = \int\! \rmd x\, \omega_\eps(x,t)\,, \qquad M_2  = \int\! \rmd x \, x_2^2\, \omega_\eps(x,t)\,.
\]

In the present case, since the vector field $\bs F^\eps =(F^\eps_z,F^\eps_r,F^\eps_\theta) := (F^\eps_1,F^\eps_2,0)$ has zero divergence, the conservation law of $M_0$ is still valid by Liouville's theorem and Eq.~\eqref{cons-omr},
\begin{equation}
\label{m0c}
M_0(t) = \frac{1}{2\pi} \int\! \rmd\bs\xi\, \frac{\omega_\eps(\bs\xi,t)}{r} =  \int\! \rmd\bs\xi_0 \, \frac{\omega_\eps(\bs\xi_0,0)}{r_0} = M_0(0)
\end{equation}
(above, the change of variables is $\bs\xi=\phi^t(\bs\xi_0)$ with $\phi^t$ the flow generated by $\dot{\bs\xi} = \bs u(\bs\xi,t) + \bs F^\eps(\bs\xi,t)$).

Since $\omega_\eps(x,t)$ has compact support, the time derivative of $M_2$ can be computed using Eq.~\eqref{weqF} with $f(x,t) = x_2^2$, hence
\[
\dot M_2 = \int\! \rmd x \, \omega_\eps(x,t) \,2 x_2  F^\eps_2(x,t) ,
\]
which implies $|\dot M_2| \le 2C_F |\log\eps|^{-3/2} \sqrt{M_2}$ in view of  item (b) in Assumption \ref{ass:1}. Therefore, 
\[
M_2 \le 2(|\zeta^0_2| + \eps)^2 M_0 + 32C_F^2T^2 |\log\eps|^{-3} \le \const |\log\eps|^{-1}\,,
\] 
where we used also Eq.~\eqref{initialF}. This is the same estimate, but for a larger constant, which is obtained in absence of $F^\eps$, and the particular value of this constant is easily seen to be irrelevant in the proof of \cite[Thm.~1]{BCM00}.

Concerning the variation of $E$, in the proof of \cite[Thm.~1]{BCM00}  the energy conservation is used to deduce a lower bound of the form $E>C^*/|\log\eps| + \const /|\log\eps|^2$, because such estimate can be easily shown to be true for the energy computed at time zero. The crucial point in the argument is that $C^*$ is the same constant appearing in the leading term of an upper bound deduced for the energy $E$ at any time. Therefore, the argument given in \cite{BCM00} applies also here if we show that $|\dot E| \le \const/|\log\eps|^2$.   

To compute $\dot E$, we introduce the stream function
\[
\Psi(x,t) = \int\!\rmd y\, S(x,y)\, \omega_\eps(y,t)\,, 
\]
where
\[
S(x,y) := \frac{x_2y_2}{2\pi} \int_0^\pi\!\rmd\theta\, \frac{\cos\theta}{\sqrt{|x-y|^2 + 2x_2y_2(1-\cos\theta)}}\,,
\]
so that $u(x,t) = x_2^{-1} \nabla^\perp\Psi(x,t)$ and the energy reads (see, e.g., \cite{BCM00,F})
\[
E = \pi \int\! \rmd x\, \Psi(x,t) \, \omega_\eps(x,t) = \pi \int\!\rmd x \int\!\rmd y\, S(x,y)\, \omega_\eps(x,t) \, \omega_\eps(y,t)\,.
\]
Since $\omega_\eps(x,t)$ has compact support we can apply Eq.~\eqref{weqF}, so that, since $u\cdot \nabla \Psi=0$,
\[
\dot E = \pi \int\!\rmd x\, \omega_\eps(x,t) \, (F^\eps \cdot \nabla \Psi + \partial_t \Psi)(x,t)\,,
\]
where, using again Eq.~\eqref{weqF},
\[
\partial_t\Psi(x,t) = \int\!\rmd y\, \omega_\eps(y,t)\, (u+F^\eps)(y,t) \cdot \nabla_y S(x,y)\,.
\]
By the symmetry of $S(x,y)$ we have $\Psi(y,t) = \int\!\rmd x\, S(x,y)\, \omega_\eps(x,t)$, whence
\[
\begin{split}
\int\!\rmd x\, \omega_\eps(x,t)\, \partial_t \Psi (x,t) & = \int\!\rmd x\, \omega_\eps(x,t) \int\!\rmd y\, \omega_\eps(y,t)\, (u+F^\eps)(y,t) \cdot \nabla_y S(x,y) \\ & = \int\!\rmd y\, \omega_\eps(y,t)\, (u+F^\eps)(y,t) \cdot \nabla_y \int\!\rmd x\,S(x,y)\, \omega_\eps(x,t)\\ &  = \int\!\rmd y\, \omega_\eps(y,t) \, (F^\eps \cdot \nabla \Psi)(y,t)\,, 
\end{split} 
\]
where we used again the orthogonality condition $u\cdot \nabla \Psi=0$. In conclusion,
\[
\begin{split}
\dot E & = 2 \pi \int\!\rmd x\, \omega_\eps(x,t) (F^\eps \cdot \nabla \Psi)(x,t) \\ & = \ 2\pi \int\!\rmd x \!\int\! \rmd y\, \omega_\eps(x,t)\, F^\eps(x,t)\cdot  \nabla_xS(x,y)\, \omega_\eps(y,t)\,.
\end{split}
\]
Recalling that $\nabla^\perp\Psi(x,t) = x_2 u(x,t)$, from Eqs.~\eqref{u=}, \eqref{H1}, and \eqref{H2}, we have 
\[
\begin{split}
\frac{\partial S}{\partial x_1} & = \frac{x_2y_2}{2\pi} \int_0^\pi\!\rmd\theta\, \frac{(y_1-x_1)\cos\theta}{[|x-y|^2 + 2x_2y_2(1-\cos\theta)]^{3/2}}\,, \\ \frac{\partial S}{\partial x_2} & = \frac{x_2y_2}{2\pi} \int_0^\pi\!\rmd\theta\, \frac{y_2-x_2\cos\theta}{[|x-y|^2 + 2x_2y_2(1-\cos\theta)]^{3/2}}\,.
\end{split}
\]
Therefore, recalling the definitions \eqref{ar} and \eqref{I1I2},
\[
\nabla_xS(x,y) = \frac{I_1(a)(y-x)}{2\pi\sqrt{x_2y_2}}  + \frac{I_2(a)}{2\pi}\sqrt{\frac{y_2}{x_2}} \begin{pmatrix} 0 \\ 1 \end{pmatrix} =: A(x,y) + B(x,y),
\]
with $A(x,y)$ anti-symmetric, so that
\[
\begin{split}
\dot E & = \pi \int\!\rmd x \!\int\! \rmd y\, \omega_\eps(x,t) [F^\eps(x,t) - F^\eps(y,t)] \cdot A (x,y) \, \omega_\eps(y,t) \\ & \quad + 2\pi \int\!\rmd x \!\int\! \rmd y\, \omega_\eps(x,t) F^\eps(x,t)\cdot B(x,y) \, \omega_\eps(y,t)\,.
\end{split}
\]
In view of item (b) in Assumption \ref{ass:1}, as $\div(x_2 F^\eps) =0$ implies $F^\eps_2((x_1,0),t)=0$, we deduce that
\[
\begin{split}
& \big|(F^\eps(x,t) - F^\eps(y,t)) \cdot A (x,y)\big| \le \frac{L}{2\pi|\log\eps|} I_1(a) \frac{|y-x|^2 }{\sqrt{x_2y_2}} \,, \\ & \big| F^\eps(x,t)\cdot B(x,y) \big| = \big| (F^\eps(x,t) - F^\eps(x_1,0,t)) \cdot B(x,y) \big| \le \frac{L}{2\pi|\log\eps|} I_2(a) \sqrt{x_2y_2}\,. 
\end{split}
\]
We now use the following estimates, see the Appendix in \cite{Mar99},
\[
\begin{split}
I_1(a) & \le \int_0^\pi\!\rmd\theta\, \frac{\cos\frac\theta 2}{[a^2+2(1-\cos\theta)]^{3/2}} = \int_0^\pi\!\rmd\theta\, \frac{\cos\frac\theta 2}{[a^2+4\sin^2\frac\theta 2]^{3/2}} = \frac{2}{a^2\sqrt{a^2+4}}\,, \\ I_2(a) & \le \frac 12\log\left(2+\sqrt{a^2+4}\right) - \frac 12 \log a + \const \\ & = \frac 12 \log\left(2\sqrt{x_2y_2} + \sqrt{b+4x_2y_2}\right) -\frac 14 \log b + \const,
\end{split}
\]
where in the last identity we introduced the quantity $b=|x-y|^2$ as in \cite{BCM00}. Therefore,
\[
\big| (F^\eps(x,t) - F^\eps(y,t)) \cdot A (x,y)\big| \le  \frac{\const}{|\log\eps|}\sqrt{x_2y_2}\,,
\]
\[
\begin{split}
\big| F^\eps(x,t)\cdot B(x,y) \big| & \le \frac{\const}{|\log\eps|} \\ & \times \sqrt{x_2y_2} \left[ \log\left(2\sqrt{x_2y_2} + \sqrt{b+4x_2y_2}\right) -\frac 12 \log b + \const \right],
\end{split}
\]
so that 
\[
\begin{split}
|\dot E| & \le \frac{\const}{|\log\eps|} \int\!\rmd x \!\int\! \rmd y\, \omega_\eps(x,t) \omega_\eps(y,t) \\ & \quad \times \sqrt{x_2y_2} \left\{\const + \log\left(2\sqrt{x_2y_2} + \sqrt{b+4x_2y_2}\right) -\frac 12 \log b \right\}.  
\end{split}
\]
Apart from the factor $\const/|\log\eps|$, the right-hand side in the above inequality coincides with the first upper bound on $E$ appearing in \cite[Eq.\ (2.30)]{BCM00}, from which the estimate $E \le \const/|\log\eps|$ is deduced. We conclude that $|\dot E| \le \const / |\log\eps|^2$ as required.
\end{proof}

\end{document}